\documentclass[12pt,a4wide]{article}
\pdfoutput=1
\usepackage{amssymb,a4wide}
\usepackage{amsmath}
\usepackage{latexsym,url,xspace}
\usepackage{extarrows}

\newtheorem{theorem}{Theorem}

\newtheorem{corollary}[theorem]{Corollary}

\newtheorem{proposition}[theorem]{Proposition}

\newcommand{\goitN}{\ensuremath{\mathcal{G}_{OIT}^{N}\xspace}}

\newtheorem{example}[theorem]{Example}
\newenvironment{proof}[1][Proof]{\textbf{#1.} }{\ \rule{0.5em}{0.5em}}

  \newcommand{\cut}[1]{}

\begin{document}

\title{Transformations of normal form games \\ by preplay offers for payments among players}

\author{Valentin Goranko \\ Technical University of Denmark,  email: \texttt{vfgo@imm.dtu.dk}} 


\date{August 6, 2012}
\maketitle

\begin{abstract}
We consider transformations of normal form games by binding preplay offers of players for payments of utility to other players conditional on them playing designated in the offers strategies. The game-theoretic effect of such preplay offers is  transformation of the payoff matrix of the game by transferring payoffs between players. Here we analyze and completely characterize the possible transformations of the payoff matrix of a normal form game by sets of preplay offers. 
\end{abstract}

\section{Introduction: the conceptual basis}

It is well known that some normal form games have no pure strategy Nash equilibria, while others, like the Prisoners' Dilemma, have rather unsatisfactory -- e.g., strongly Pareto dominated -- ones. Sometimes, mutually more beneficial outcomes can be achieved if players could communicate and make \emph{binding offers} for payments of bonuses to other players before the play of the game in order to provide additional incentives for them to play desired strategies. More precisely, we assume the possibility that: 

\begin{quote}
 \emph{Before playing the game any player $A$ can make a binding offer to any other player $B$ to pay him, after the game is played, a declared amount of utility $\delta$ if $B$ plays a strategy $s$ specified in the offer by $A$.}
\end{quote}

\newpage

Here is our basic assumption in more details: 
\begin{itemize}
\item Any preplay offer of a player A is \emph{binding and irrevocable for $A$}, and only contingent on B playing the strategy $s$ specified by A. 

\item  However, such offer \emph{does not} create any obligation for $B$ (and, therefore, it does not transform the game into a cooperative one), as $B$ is still at liberty to choose his strategy when the game is actually played. 
 
\item Offers can neither be withdrawn nor rejected.  However, as we will show, they can be effectively cancelled by a suitable exchange of 'counter-offers' by both players involved. 
 
\item Offers can only be made for non-negative payments. 
Again, we will show that negative offers (regarded as threats for punishment) can be simulated, 
 but only by cooperation of both players involved.   
 
\item 
 In principle, preplay offers are unbounded. 
In reality, offers of rational players are bounded above by the currently maximal payoff in the game. 

\item We only consider offers contingent on pure strategies, even though players can still play mixed strategies. 
\end{itemize}

Every player can make several offers, to different players, so the possible behaviours of the players remain, in principle, complex and unconstrained. 

\medskip
A key observation:  \emph{every preplay offer transforms the normal form game into another one, by an explicitly defined transformation of the payoff matrix}. 

\medskip
In this paper we study the purely mathematical effect of preplay offers on the payoff matrices of normal form games, disregarding any rationality considerations that may prescribe to players if and what offers to make. That is, here we are only interested in \emph{how} a game matrix can be transformed by preplay offers, but not \emph{why} players may wish to exchange offers in order to effect a given possible transformation. The latter, which is the truly game theoretic question, we study in the separate papers \cite{manifesto,Preplay-tech1}. 

\medskip 
The contributions of the present paper are technical: we characterize completely and rather transparently the game matrix transformations that can be induced by preplay offers of the type described above. 

\medskip
We note that somewhat more general types of side payments -- not only positive but also possibly negative (threats for punishments) and contingent not just on the recipient's strategy but also on own actions, or on an entire strategy profile (i.e., on an outcome) -- have been studied before in the literature, most notably by Jackson and Wilkie \cite{JW05} and the recent follow-up by Ellingsen and Paltseva \cite{EP11}. There are some essential differences in the assumptions made in those papers and in the present work, and respectively in the game-theoretic properties and effect of such payments, as demonstrated in \cite{manifesto} and \cite{Preplay-tech1}. As we show in the present paper, the game matrix transformations induced by preplay offers considered here can simulate negative payments, but not the offers for payments contingent on strategy profiles considered in  \cite{JW05}  and \cite{EP11}.

\section{Preplay offers and payoff matrix transformations}

\subsection{A motivating example}

Consider a standard version of the Prisoners' Dilemma (PD)  game with a payoff matrix  
\[
\left.\begin{array}{c|c|c}
I \backslash II & C & D \\
\hline 
C & 4,4 & 0,5  \\
\hline 
D & 5,0 & 1,1
\end{array}\right.
\]

The only Nash Equilibrium is (D,D) with the paltry payoff (1,1).

Now, suppose I makes to II a \emph{binding offer} to pay her 2 utils after the game if II plays C. That offer transforms the game by transferring 2 utils from the payoff of I to the payoff of II in every entry of the column of the game matix, where II plays C, as follows: 
\[
\left.\begin{array}{c|c|c}
I \backslash II & C & D \\
\hline 
C & 2,6 & 0,5  \\
\hline 
D & 3,2 & 1,1
\end{array}\right.
\]

In this game, player I still has the incentive to play D, which strictly dominates C for him, but the dominant strategy for II now is C, and thus the only Nash Equilibrium is (D,C) with payoff (3,2) -- strictly better than the original payoff (1,1).

Of course, II can now realize that player I has no incentive to cooperate yet. That incentive, however, can be created if player II, too, makes an offer to pay  2 utils to player I after the game, if I cooperates. Then the game transforms as follows:
\[
\left.\begin{array}{c|c|c}
I \backslash II & C & D \\
\hline 
C & 4,4 & 2,3  \\
\hline 
D & 3,2 & 1,1
\end{array}\right.
\]

In this game, the only Nash Equilibrium is (C,C) with payoff (4,4), which is Pareto optimal.
Note that this is the same payoff for (C,C) as in the original PD game, but now both players have transformed he game into one where they both have incentives to cooperate, and have thus escaped from the trap of the original Nash Equilibrium (D,D).

\subsection{Transformations of normal form games by preplay offers}

Now, let us generalize. 
Consider a general 2-person (for technical simplicity only) normal form (NF) game with a payoff matrix 
\[
\left.\begin{array}{c|c|c|c|c}
A \backslash B  & B_{1} &  \cdots &  B_{j} & \cdots \\\hline
A_{1} & \cdots & \cdots & a_{1j}, b_{1j} & \cdots  \\\hline 
A_{2} & \cdots & \cdots & a_{2j}, b_{2j} & \cdots  \\\hline 
  \cdots &  \cdots &  \cdots &  \cdots &  \cdots \\\hline 
A_{i} & \cdots & \cdots & a_{ij}, b_{ij} & \cdots  \\\hline 
  \cdots &  \cdots &  \cdots &  \cdots &  \cdots \\\hline  
\end{array}\right.
\]

Suppose player A makes a preplay offer to player B to pay her additional utility $\delta\geq 0$
\footnote{Of course, preplay offers where $\delta= 0$ make no difference in the game. 
The technical reason we allow $\delta= 0$ is to have an identity transformation at hand, but such vacuous offers can also be used by players for signalling.} 
 if $B$ plays $B_{j}$. We will denote such offer by 
 $A \xlongrightarrow{\delta \slash  B_{j}} B$.

That offer transforms the payoff matrix of the game as follows: 
\[
\left.\begin{array}{c|c|c|c|c}
A \backslash B  & B_{1} &  \cdots &  B_{j} & \cdots \\\hline
A_{1} & \cdots & \cdots & a_{1j}-\delta, b_{1j}+\delta & \cdots  \\\hline 
A_{2} & \cdots & \cdots & a_{2j}-\delta, b_{2j}+\delta & \cdots  \\\hline 
  \cdots &  \cdots &  \cdots &  \cdots &  \cdots \\\hline 
A_{i} & \cdots & \cdots & a_{ij}-\delta, b_{ij}+\delta & \cdots  \\\hline 
  \cdots &  \cdots &  \cdots &  \cdots &  \cdots \\\hline  
\end{array}\right.
\]

We will call such transformation of a payoff matrix a \emph{primitive offer-induced transformation}, or a POI-transformation, for short. 

Several preplay offers can be made by each player. Clearly, the transformation of a payoff matrix induced by several preplay offers can be obtained by applying the POI-transformations corresponding to each of the offers consecutively, in any order. 
We will call such transformations \emph{offer-induced transformations}, or OI-transformations, for short. Thus, every OI-transformation corresponds to a \emph{set} of preplay offers, respectively a set of POI-transformations. Note that the sets generating a given 
OI-transformation need not be unique. For instance, A can make two independent consecutive offers 

$A \xlongrightarrow{\delta_{1} \slash  B_{j}} B$
and 
$A \xlongrightarrow{\delta_{2} \slash  B_{j}} B$
equivalent to the offer 
$A \xlongrightarrow{\delta_{1}+\delta_{2} \slash  B_{j}} B$.

Thus, every OI-transformation has a canonical form 

\[\{A \xlongrightarrow{\delta^{A}_{1} \slash  B_{1}} B, \ldots,  
A \xlongrightarrow{\delta^{A}_{m} \slash  B_{m}}  B \} \cup 
\{B \xlongrightarrow{\delta^{B}_{1} \slash  A_{1}} A, \ldots,  
B \xlongrightarrow{\delta^{B}_{n} \slash  A_{n}} A \}\]
for some non-negative numbers $\delta^{A}_{1}, \ldots \delta^{A}_{m} ,\delta^{B}_{1}, \ldots \delta^{B}_{n}$. 

That transformation changes the payoffs from  $(a_{ij},b_{ij})$  to  $(\hat{a}_{ij},\hat{b}_{ij})$ as follows: 
\[ \hat{a}_{ij} = a_{ij} - \delta^{A}_{j} +  \delta^{B}_{i}, \ \  \hat{b}_{ij} = b_{ij} + \delta^{A}_{j} -  \delta^{B}_{i}.\]

Note that the players can collude to make \emph{any} outcome, with \emph{any} non-negative re-distribution of the payoffs in it, a \emph{strictly dominant strategy equilibrium}, by exchanging sufficiently hight offers to make the strategies generating that outcome strictly dominant\footnote{Clearly, rational players would only be interested in making offers inducing payoffs that are optimal for them. Thus, a `preplay negotiation phase' emerges which is studied in \cite{manifesto,Preplay-tech1}; here we do not take rationality considerations into account.}. 

\section{The group of offer-induced game transformations}

We begin with some useful general observations on OI-transformations of NF games for any fixed number of players $N$.

\begin{enumerate}
\item An OI-transformation does not change the sum of the payoffs of all players in any outcome, only re-distributes them. In particular, OI-transformations preserve the class of zero-sum games.

\item An OI-transformation induced by a preplay offer of player A does not change the preferences of A regarding his own strategies. In particular, (weak or strict) dominance between strategies of player  A is invariant under OI-transformations induced by preplay offers of A, i.e.: a strategy $A_{i}$ dominates (weakly, resp. strongly) a strategy $A_{j}$ before a transformation induced by a preplay offer made by A if and only if $A_{i}$ dominates (weakly, resp. strongly) $A_{j}$ after the transformation.  

\end{enumerate}

\begin{proposition}
The set of all OI-transformations of payoff matrices of $N$-person strategic games, for any fixed $N>$1, forms a commutative group under composition. 
\end{proposition}

\begin{proof} 
The composition of two OI-transformations is an  OI-transformation, corresponding to the union of the preplay offers generating the two transformations. Furthermore, the composition of OI-transformations is clearly associative and commutative, because the order of transforming the matrix with regards to the primitive preplay offers generating the composed OI-transformations is not essential.  

The  OI-transformation corresponding to any offer 
$A \xlongrightarrow{0 \slash  B_{j}} B$
is the identity transformation.

The existence  of an inverse OI-transformation to any OI-transformation is a bit trickier, because an offered payment cannot be negative. However, note first that the inverse of the POI-transformation corresponding to an offer 
$A \xlongrightarrow{\delta \slash  B_{j}} B$
can be composed from the following offers: 
\begin{itemize}
\item Player A makes an offer 
$A \xlongrightarrow{\delta \slash  B_{k}} B$
for every strategy $B_{k}$  of B, for $k \neq j$.  That basically means that A offers a reward $\delta$ to B if B \emph{does not} play the strategy  $B_{j}$. 
\item Player B makes an offer 
$B \xlongrightarrow{\delta \slash  A_{i}} A$
for \emph{every} strategy $A_{i}$ of A. That basically means that B offers unconditionally a refund $\delta$ to A.
\end{itemize}

The cumulative effect of these offers is that none of A and B gets anything from the other if 
B plays any $B_{k}$ where $k \neq j$, but if B plays $B_{j}$ then she effectively pays back 
an amount $\delta$ to A, thus canceling the offer 
$A \xlongrightarrow{\delta \slash  B_{j}} B$

Finally, the inverse of any OI-transformation $\mathcal{T}$ can be obtained by composing the inverses of the primitive OI-transformation of which $\mathcal{T}$ is composed. 
\end{proof} 

\medskip

The above proof also shows that offers for negative payments (i.e. threats of punishments) can be effected by OI-transformations, too. Thus, henceforth we may assume that an offered payment may be any real number, and an OI-transformation has a canonical form 

\[\{A \xlongrightarrow{\delta^{A}_{1} \slash  B_{1}} B, \ldots,  
A \xlongrightarrow{\delta^{A}_{m} \slash  B_{m}}  B \} \cup 
\{B \xlongrightarrow{\delta^{B}_{1} \slash  A_{1}} A, \ldots,  
B \xlongrightarrow{\delta^{B}_{n} \slash  A_{n}} A \}\]
for some \emph{real} numbers $\delta^{A}_{1}, \ldots \delta^{A}_{m} ,\delta^{B}_{1}, \ldots \delta^{B}_{n}$.  

\medskip
Furthermore, OI-transformations can also simulate the more complex offers, such as the conditional offers considered in \cite{manifesto}. 

On the other hand, it will follow form Theorem \ref{GT2} that OI-transformations cannot simulate the offers considered in \cite{JW05} and \cite{EP11}, contingent not just on the recipient's strategy but on an entire strategy profile (i.e., on an outcome). 

\medskip
We denote the group of OI-transformations of $N$-person NF games by \goitN. 
For a given payoff matrix $M$ we denote by $\goitN(M)$ the \emph{orbit of $M$ under $\goitN$}, i.e.,  the result of the action of that group on $M$, which is the set of all payoff matrices obtained by applying OI-transformations 
to $M$. We will also call $\goitN(M)$ the 
 \emph{OIT type of $M$}. Note that, because $\goitN$ is a group, the set of OIT types of NF game matrices forms a partition of the set of NF game matrices and thus generates an equivalence relation on that set. In particular, every payoff matrix in $\goitN(M)$ has the  OIT game type of $M$.

\section{Characterizing the OI-transformations of 2-person normal form games}

Here we study and answer  the question: given a NF game matrix $M$, what are the possible results of OI-transformations of $\mathcal{M}$? That is, how is OIT type of $M$ characterized and constructed? 

Let us first re-phrase the question: given another game matrix  $\widehat{M}$ of the same dimensions, \emph{can $\widehat{M}$ be obtained from $M$ by an OI-transformation}? 

To answer this question, let us first consider, for technical simplicity, the case of 2-person games. Let 

\[ M = 
\left.\begin{array}{c|c|c|c|c|c|}
A \backslash B  & B_{1} &  \cdots &  B_{j} & \cdots & B_{m}  \\\hline
A_{1} & a_{11}, b_{11}  & \cdots & a_{1j}, b_{1j} & \cdots & a_{1m}, b_{1m}  \\\hline 
  \cdots &  \cdots &  \cdots &  \cdots &  \cdots & \cdots \\\hline 
A_{i} & a_{i1}, b_{i1}  & \cdots & a_{ij}, b_{ij} & \cdots & a_{im}, b_{im}  \\\hline 
  \cdots &  \cdots &  \cdots &  \cdots &  \cdots & \cdots \\\hline  
A_{n} & a_{n1}, b_{n1}  & \cdots & a_{nj}, b_{nj} & \cdots & a_{nm}, b_{nm}  \\\hline 
\end{array}\right.
\]
Further we will use a compact notation for $M$ as follows: 
$M = \bigg[a_{ij}, b_{ij}\bigg]_{{i = 1,\ldots, n}\atop{j = 1,\ldots, m}}$.

\medskip
Now, let 
\[ \widehat{M} = 
\left.\begin{array}{c|c|c|c|c|c|}
A \backslash B  & B_{1} &  \cdots &  B_{j} & \cdots & B_{m}  \\\hline
A_{1} & \hat{a}_{11}, \hat{b}_{11}  & \cdots & \hat{a}_{1j}, \hat{b}_{1j} & \cdots & \hat{a}_{1m}, \hat{b}_{1m}  \\\hline 
  \cdots &  \cdots &  \cdots &  \cdots &  \cdots & \cdots \\\hline 
A_{i} & \hat{a}_{i1}, \hat{b}_{i1}  & \cdots & \hat{a}_{ij}, \hat{b}_{ij} & \cdots & \hat{a}_{im}, \hat{b}_{im}  \\\hline 
  \cdots &  \cdots &  \cdots &  \cdots &  \cdots & \cdots \\\hline  
A_{n} & \hat{a}_{n1}, \hat{b}_{n1}  & \cdots & \hat{a}_{nj}, \hat{b}_{nj} & \cdots & \hat{a}_{nm}, \hat{b}_{nm}  \\\hline 
\end{array}\right.
\]

Suppose $\widehat{M}$ can be obtained from $M$ by an OI-transformation
\[\tau = 
\{A \xlongrightarrow{\delta^{A}_{1} \slash  B_{1}} B, \ldots,  
A \xlongrightarrow{\delta^{A}_{m} \slash  B_{m}}  B \} \cup 
\{B \xlongrightarrow{\delta^{B}_{1} \slash  A_{1}} A, \ldots,  
B \xlongrightarrow{\delta^{B}_{n} \slash  A_{n}} A \}\]
for some real numbers $\delta^{A}_{1}, \ldots \delta^{A}_{m} ,\delta^{B}_{1}, \ldots \delta^{B}_{n}$.

Recall, that the transformation $\tau$ changes the payoffs as follows: 
\[ \hat{a}_{ij} = a_{ij}  +  \delta^{B}_{i} - \delta^{A}_{j}, \ \  \hat{b}_{ij} = b_{ij} + \delta^{A}_{j} -  \delta^{B}_{i}.\]

\begin{theorem} 
\label{GT2}
Let $M = \bigg[a_{ij}, b_{ij}\bigg]_{{i = 1,\ldots, n}\atop{j = 1,\ldots, m}} $ and  $\widehat{M} = \bigg[\hat{a}_{ij}, \hat{b}_{ij}\bigg]_{{i = 1,\ldots, n}\atop{j = 1,\ldots, m}} $ be $2$-person NF game matrices  of dimensions $n\times m$. The matrix $\widehat{M}$ can be obtained from $M$ by an OI-transformation if and only if the following two conditions hold, where  $c_{ij} = \hat{a}_{ij} - a_{ij}$ and $d_{ij} = \hat{b}_{ij} - b_{ij}$: 

\begin{enumerate}
\item[$C_{1}$:] $a_{ij} + b_{ij} = \hat{a}_{ij} + \hat{b}_{ij}$, 
or equivalently, $d_{ij} = -c_{ij}$, 
for all $i = 1,\ldots n, j = 1,\ldots m$.  
\item[$C_{2}$:] $c_{ij} - c_{i(j+1)} = c_{(i+1)j} - c_{(i+1)(j+1)}$, or equivalently, 
$c_{ij} + c_{(i+1)(j+1)} = c_{i(j+1)} + c_{(i+1)j}$,
for all $i = 1,\ldots n-1, j = 1,\ldots m-1$.  
\end{enumerate}
\end{theorem} 

\begin{proof}
Condition $C_{1}$ is  obviously necessary, because all side payments are between the two players, so the sum of their payoffs in any given outcome remains constant. 

To show the necessity of $C_{2}$, a simple calculation suffice: 

$c_{ij} + c_{(i+1)(j+1)}$  = 

$\hat{a}_{ij} - a_{ij} + \hat{a}_{(i+1)(j+1)} - a_{(i+1)(j+1)}$ = 

$a_{ij}  +  \delta^{B}_{i} - \delta^{A}_{j} - a_{ij} + a_{(i+1)(j+1)} 
 +  \delta^{B}_{i+1} - \delta^{A}_{j+1} - a_{(i+1)(j+1)}$ = 
 
$\delta^{B}_{i} - \delta^{A}_{j}  +  \delta^{B}_{i+1} - \delta^{A}_{j+1}$. 

\medskip 
Likewise:  

$ c_{i(j+1)} + c_{(i+1)j}$  = 

$\hat{a}_{i(j+1)} - a_{i(j+1)} + \hat{a}_{(i+1)j} - a_{(i+1)j}$ = 

$a_{i(j+1)}  +  \delta^{B}_{i} - \delta^{A}_{j+1} - a_{i(j+1)} + a_{(i+1)j} 
 +  \delta^{B}_{i+1} - \delta^{A}_{j} - a_{(i+1)j}$ = 
 
$\delta^{B}_{i} - \delta^{A}_{j+1} +  \delta^{B}_{i+1} - \delta^{A}_{j}$. 

\medskip 
Clearly, the two results are equal.

For the sufficiency, suppose $C_{1}$ and $C_{2}$ hold. Then, first observe that, due to $C_{1}$, any OI-transformation that transforms A's payoffs in $M$ into A's payoffs in $\widehat{M}$ will transform accordingly the payoffs of B in $M$ into the payoffs of B in $\widehat{M}$. So, we can ignore B's payoffs and consider only the transformation of the matrices of A's payoffs. 

In order to prove the existence of (real valued) payments $\delta^{A}_{1}, \ldots \delta^{A}_{m} ,\delta^{B}_{1}, \ldots \delta^{B}_{n}$ that effect the transformation from $M$ to  $\widehat{M}$ ,  
we consider the system of  $mn$ linear equations for these $n+m$ unknowns that expresses the changes of A's payoffs: 
\[ \| \delta^{B}_{i} - \delta^{A}_{j} = c_{ij}, \ \ \mbox{for all} \ i = 1,\ldots n, j = 1,\ldots m. \]

Thus, $\widehat{M}$ can be obtained from $M$ by an OI-transformation precisely when that system has a real solution. The rest is application of standard linear algebra. The matrix of the system is 

\[
\left(\begin{array}{cccccccc|c} \delta^{B}_{1} &  \delta^{B}_{2} & \ldots & \delta^{B}_{n} & \delta^{A}_{1}  & \delta^{A}_{2} & \ldots & \delta^{A}_{m}  & c_{ij} \\
\hline 
1 & 0 & 0 & 0 & -1 & 0 & 0 & 0  & c_{11} \\
1 & 0 & 0 & 0 & 0 & -1 & 0 & 0  & c_{12} \\
\ldots & \ldots  &  \ldots & \ldots & \ldots & \ldots & \ldots & \ldots  & \ldots \\
1 & 0 &  0 & 0 & 0 & 0 & 0 & -1  & c_{1m} \\
0 & 1 & 0 & 0 & -1 & 0 & 0 & 0  & c_{21} \\
0 & 1 &  0 & 0 & 0 & -1 & 0 & 0  & c_{22} \\
\ldots & \ldots  &  \ldots & \ldots & \ldots & \ldots & \ldots & \ldots  & \ldots \\
0 & 1 & 0 & 0 & 0 & 0 & 0 & -1  & c_{2m} \\
\ldots & \ldots  &  \ldots & \ldots & \ldots & \ldots & \ldots & \ldots  & \ldots \\
\ldots & \ldots  &  \ldots & \ldots & \ldots & \ldots & \ldots & \ldots  & \ldots \\
0 & 0 & 0 & 1 & -1 & 0 & 0 & 0  & c_{n1} \\
0 & 0 &  0 & 1 & 0 & -1 & 0 & 0  & c_{n2} \\
\ldots & \ldots  &  \ldots & \ldots & \ldots & \ldots & \ldots & \ldots  & \ldots \\
0 & 0 & 0 & 1 & 0 & 0 & 0 & -1  & c_{nm} \\
\end{array}\right)
\]

We now apply to it the Gauss elimination method. Subtracting row 1 from each of rows $2,\ldots,m$, then row $m+1$ from rows 
$m+2,\ldots,2m$, etc., and finally row $(n-1)m+1$ from rows $(n-1)m+2,\ldots,nm$ produces: 

\[
\left(\begin{array}{cccccccc|c} \delta^{B}_{1} &  \delta^{B}_{2} & \ldots & \delta^{B}_{n} & \delta^{A}_{1}  & \delta^{A}_{2} & \ldots & \delta^{A}_{m} & c_{ij} \\
\hline 
1 & 0 & 0 & 0 & -1 & 0 & 0 & 0  & c_{11} \\
0 & 0 & 0 & 0 & 1 & -1 & 0 & 0  & c_{12}-c_{11} \\
\ldots & \ldots  &  \ldots & \ldots & \ldots & \ldots & \ldots & \ldots  & \ldots \\
0 & 0 &  0 & 0 & 1 & 0 & 0 & -1  & c_{1m}-c_{11} \\
0 & 1 & 0 & 0 & -1 & 0 & 0 & 0  & c_{21} \\
0 & 0 &  0 & 0 & 1 & -1 & 0 & 0  & c_{22}-c_{21} \\
\ldots & \ldots  &  \ldots & \ldots & \ldots & \ldots & \ldots & \ldots  & \ldots \\
0 & 0 & 0 & 0 & 1 & 0 & 0 & -1  & c_{2m}-c_{21} \\
\ldots & \ldots  &  \ldots & \ldots & \ldots & \ldots & \ldots & \ldots  & \ldots \\
\ldots & \ldots  &  \ldots & \ldots & \ldots & \ldots & \ldots & \ldots  & \ldots \\
0 & 0 & 0 & 1 & -1 & 0 & 0 & 0  & c_{n1}\\
0 & 0 &  0 & 1 & 0 & -1 & 0 & 0  & c_{n2}-c_{n1}  \\
\ldots & \ldots  &  \ldots & \ldots & \ldots & \ldots & \ldots & \ldots  & \ldots \\
0 & 0 & 0 & 1 & 0 & 0 & 0 & -1  & c_{nm}-c_{n1}  \\
\end{array}\right)
\]

Now, note that rows $k$, $m+k$, \ldots $(n-1)m+k$ have the same left hand sides, for each $k=2,3,\ldots,m$. For the system to be consistent, the right hand sides must be equal, too. Indeed, by consecutive applications of condition $C_{2}$: $c_{1k}-c_{11}$ = $c_{2k}-c_{21}$ = \ldots $c_{nk}-c_{n1}$.

Further, we subtract row $k$ from each of $m+k$, \ldots $(n-1)m+k$, for each $k=2,3,\ldots,m$ and remove the resulting 0-rows. After re-arrangement of the remaining rows we obtain: 

\[
\left(\begin{array}{cccccccc|c} \delta^{B}_{1} &  \delta^{B}_{2} & \ldots & \delta^{B}_{n} & \delta^{A}_{1}  & \delta^{A}_{2} & \ldots & \delta^{A}_{m} & c_{ij} \\
\hline 
1 & 0 & 0 & 0 & -1 & 0 & 0 & 0  & c_{11} \\
0 & 1 & 0 & 0 & -1 & 0 & 0 & 0  & c_{21} \\
\ldots & \ldots  &  \ldots & \ldots & \ldots & \ldots & \ldots & \ldots  & \ldots \\
0 & 0 & 0 & 1 & -1 & 0 & 0 & 0  & c_{n1}\\
0 & 0 & 0 & 0 & 1 & -1 & 0 & 0  & c_{12}-c_{11} \\
\ldots & \ldots  &  \ldots & \ldots & \ldots & \ldots & \ldots & \ldots  & \ldots \\
0 & 0 &  0 & 0 & 1 & 0 & 0 & -1  & c_{1m}-c_{11} \\
\end{array}\right)
\]

Finally, subtracting row $n+m-2$ from $n+m-1$, then row  $n+m-3$ from $n+m-2$, etc, and lastly, row $n+1$ from $n+2$, we obtain:  

\[
\left(\begin{array}{cccccccc|c} \delta^{B}_{1} &  \delta^{B}_{2} & \ldots & \delta^{B}_{n} & \delta^{A}_{1}  & \delta^{A}_{2} & \ldots & \delta^{A}_{m} & c_{ij} \\
\hline 
1 & 0 & 0 & 0 & -1 & 0 & 0 & 0  & c_{11} \\
0 & 1 & 0 & 0 & -1 & 0 & 0 & 0  & c_{21} \\
\ldots & \ldots  &  \ldots & \ldots & \ldots & \ldots & \ldots & \ldots  & \ldots \\
0 & 0 & 0 & 1 & -1 & 0 & 0 & 0  & c_{n1}\\
0 & 0 & 0 & 0 & 1 & -1 & 0 & 0  & c_{12}-c_{11} \\
0 & 0 & 0 & 0 & 0 & 1 & -1 & 0  & c_{13}-c_{12} \\
\ldots & \ldots  &  \ldots & \ldots & \ldots & \ldots & \ldots & \ldots  & \ldots \\
0 & 0 &  0 & 0 & 0 & 0 & 1 & -1  & c_{1m}-c_{1(m-1)} \\
\end{array}\right)
\]
The ranks of the matrix and the extended matrix above are clearly equal to $n+m-1$, so the corresponding system is consistent. Moreover, it has infinitely many solutions obtained by treating $\delta^{A}_{m}$ as a real parameter and solving for all other unknowns in terms of it. This completes the proof.  
\end{proof}

\medskip 
\begin{corollary}
In the case of $2\times 2$ payoff matrices 
\[ M = \left[\begin{array}{c|cc}
A\backslash  B & B_{1} & B_{2}  \\
\hline 
A_{1} &  a_{11}, b_{11} &  a_{12}, b_{12}  \\
A_{2} &  a_{21}, b_{21} &  a_{22}, b_{22} 
\end{array}\right], \ \ 
\widehat{M} = \left[\begin{array}{c|cc}
A\backslash  B & B_{1} & B_{2}  \\
\hline 
A_{1} &  \hat{a}_{11}, \hat{b}_{11} &  \hat{a}_{12}, \hat{b}_{12}  \\
A_{2} &  \hat{a}_{21}, \hat{b}_{21} &  \hat{a}_{22}, \hat{b}_{22} 
\end{array}\right]
\]
the matrix $\widehat{M}$ can be obtained from $M$ by an OI-transformation if and only if:  
\begin{enumerate}
\item[$C_{1}$:] $a_{ij} + b_{ij} = \hat{a}_{ij} + \hat{b}_{ij}$, 
for all $i,j \in \{1,2 \}$.  
\item[$C_{2}$:] $c_{11} + c_{22} = c_{12} + c_{21}$, where $c_{ij} = \hat{a}_{ij} - a_{ij}$.
\end{enumerate}
\end{corollary} 
 
  \medskip 
 \begin{corollary}
No OI-transformation applied to a game matrix $M$ can produce a game matrix that differs 
from $M$ in only one outcome. 
\end{corollary} 

Consequently, OI-transformations cannot simulate the offers considered in \cite{JW05} and \cite{EP11}, contingent on a single outcome. 
 
 \medskip 
\begin{example}
For example, the payoff matrix  
\[
 \left[\begin{array}{c|cc}
A\backslash  B & B_{1} & B_{2}  \\
\hline 
A_{1} &  4,4 &  0,5  \\
A_{2} &  3,0 & 1,1  
\end{array}\right]
\]
can be  OI-transformed to
\[
 \left[\begin{array}{c|cc}
A\backslash  B & B_{1} & B_{2}  \\
\hline 
A_{1} &  2,6 &  2,3  \\
A_{2} &  0,3 & 2,0  
\end{array}\right] 
\]
but not to 
\[
 \left[\begin{array}{c|cc}
A\backslash  B & B_{1} & B_{2}  \\
\hline 
A_{1} &  2,6 &  2,3  \\
A_{2} &  0,3 & 1,1  
\end{array}\right] 
\]
neither to 
\[
 \left[\begin{array}{c|cc}
A\backslash  B & B_{1} & B_{2}  \\
\hline 
A_{1} &  2,6 &  3,2  \\
A_{2} &  0,3 & 2,0  
\end{array}\right] 
\]
\end{example}

\medskip
The condition $C_{2}$ from Theorem \ref{GT2} can be rewritten as a recurrent formula 
$c_{(i+1)(j+1)} = c_{(i+1)j} +  c_{i(j+1)} -  c_{ij}$, 
which suggests that all values of $c_{ij}$, and therefore all values of $\hat{a}_{ij}$, can be computed iteratively from some initial values along one row and one column. Therefore, by using conditions $C_{1}$ and $C_{2}$, every OI-transformation can be determined \emph{locally}, by specifying the resulting payoffs for any of the two players on \emph{one row} and  \emph{one column}. In other words, fix a strategy profile in the transformed matrix, fix one of the players, and assign \emph{any} real payoffs for that player to all outcomes where at least one of the players follows his strategy from the fixed strategy profile. The resulting partly defined matrix can then be uniquely extended to one that can be obtained by an  OI-transformation from the initial matrix $M$. The following result formalizes that observation. 

\begin{theorem} 
\label{GT2a}
Let $M = \bigg[a_{ij}, b_{ij}\bigg]_{{i = 1,\ldots, n}\atop{j = 1,\ldots, m}}$. 
Then for every fixed $i \in 
\{1,\ldots, n\}$ and $j \in \{1,\ldots, m\}$, every tuple\footnote{Note that $\hat{a}_{ij}$ occurs twice in this list.}   
of $n+m-1$ reals $\hat{a}_{i1},
\ldots \hat{a}_{im},\hat{a}_{1j},
\ldots, \hat{a}_{nj}$ 
can be extended to a unique payoff matrix $\widehat{M} = \bigg[\hat{a}_{ij}, \hat{b}_{ij}\bigg]_{{i = 1,\ldots, n}\atop{j = 1,\ldots, m}} $ that can be obtained from $M$ by an OI-transformation.
\end{theorem} 

\begin{proof}

For notational simplicity, let us assume that $i=1$ and $j=1$ and that $m\leq n$. Clearly, the argument for any other combination of $i,j$ is analogous.  In order to determine the matrix 
$\widehat{M}$ it  suffices to determine the values of all $\hat{a}_{ij}$, for ${i = 1,\ldots, n}, {j = 1,\ldots, m}$ and then compute all values $\hat{b}_{ij}$ by applying condition $C_{1}$ from Theorem \ref{GT2}. 

Now, note that all values $\hat{a}_{ij}$  can be computed iteratively, step-by-step, by using the identities in $C_{2}$ of Theorem \ref{GT2}: first, compute $\hat{a}_{22}$; then,  
$\hat{a}_{23}$ and  $\hat{a}_{32}$, etc.. More precisely,  
given all values along the diagonal $\hat{a}_{1k},\hat{a}_{2(k-1)}\ldots,\hat{a}_{k1}$, for $k<n$, using  $C_{2}$ of Theorem \ref{GT2} one can compute uniquely values for $\hat{a}_{2(k+1)},\hat{a}_{2k}\ldots,\hat{a}_{(k+1)2}$. When $k$ increases between $m$ and $n$, the argument continues likewise, but for values along the diagonals $\hat{a}_{km},\ldots,\hat{a}_{2(m+k-2)}$, and then further, along the shrinking diagonals $\hat{a}_{km},\ldots,\hat{a}_{j(m+k-j)}$, until eventually $\hat{a}_{nm}$ is computed. 

The resulting matrix  $\widehat{M} = \bigg[\hat{a}_{ij}, \hat{b}_{ij}\bigg]_{{i = 1,\ldots, n}\atop{j = 1,\ldots, m}} $ satisfies the conditions $C_{1}$ and $C_{2}$ by construction. Therefore, by Theorem \ref{GT2}, $\widehat{M}$ can be obtained from $M$ by an OI-transformation. 
The uniqueness of $\widehat{M}$ follows from the construction, too. 

The case of arbitrary $i$ and $j$ is essentially the same, but the computation of the values of the $\hat{a}_{pq}$s  now propagates from $\hat{a}_{ij}$ in all 4 diagonal directions.
\end{proof}

Thus, in summary, theorems \ref{GT2} and \ref{GT2a} together say that any payoff matrix $\widehat{M}$ can be obtained  from matrix $M$ by an OI-transformation by choosing suitable payoffs in one row and one column satisfying condition $C_{1}$, and then computing the rest by using the recurrent formulae derived from condition $C_{2}$.

\begin{example}
Suppose the starting payoff matrix is 
\[ M = 
 \left[\begin{array}{c|ccc}
A\backslash  B & B_{1} & B_{2} & B_{3}  \\
\hline 
A_{1} &  4,4 &  6,2 &  0,6  \\
A_{2} &  2,6 & 1,1 &  2,2  \\ 
A_{3} &  5,0 & 0,1 &  1,5  \\ 
A_{4} &  0,0 & 2,3 &  3,0  \\ 
\end{array}\right]
\]
and row 1 and column 1 of the transformed matrix are as follows: 

\[ \widehat{M}  = 
 \left[\begin{array}{c|ccc}
A\backslash  B & B_{1} & B_{2} & B_{3}  \\
\hline 
A_{1} &  1,7 &  4,4 &  2,4  \\
A_{2} &  7,1 &        &    \\ 
A_{3} &  3,2 &        &    \\ 
A_{4} &  0,0 &        &    \\ 
\end{array}\right]
\]
The remaining entries are then computed consecutively from condition $C_{2}$ of Theorem \ref{GT2}: \\
First, 
we obtain 
$c_{22} =  c_{21} + c_{12} - c_{11}  $ = $ 5 + (-2) - (-3) $ = 6. \\
Then:  
$c_{23} =  c_{22} + c_{13} - c_{12} $ =  $ 6 + 2 - (-2)$ = 10; 
$c_{32} =  c_{31} + c_{22} - c_{21} $ = $ -2  + 6 - 5$ = -1, 
etc. Thus, the whole matrix $C = \big( c_{ij} \big)_{{i = 1,\ldots, 4}\atop{j = 1,\ldots, 3}}$ is computed: 
\[ C = 
 \left[\begin{array}{c|rrr}
i\backslash  j & 1 & 2 & 3  \\
\hline 
1 &  -3 &  -2 &  2  \\
2 &  5  &   6  &  10  \\ 
3 &  -2 &  -1  &  3  \\ 
4 &  0  &   1  &  5  \\ 
\end{array}\right]
\]
Eventually, from the definition of $c_{ij}$ and Condition $C_{1}$ of Theorem \ref{GT2}, we obtain: 
\[ \widehat{M}  = 
 \left[\begin{array}{c|ccc}
A\backslash  B & B_{1} & B_{2} & B_{3}  \\
\hline 
A_{1} &  1,7 &  4,4 &     2,4  \\
A_{2} &  7,1 &  7,-5 & 12,-8  \\ 
A_{3} &  3,2 &   -1,2   & 4,2   \\ 
A_{4} &  0,0 &   3, 2   & 8,-5   \\ 
\end{array}\right]
\]
\end{example}

\section{Characterizing the OI-transformations of $N$-person normal form games}

Generalizing these results to $N$-person NF games is relatively easy, but adds a substantial notational overhead. 

Let the players be indexed with $\{1,2,\ldots N\}$ and consider two NF game matrices of the same dimensions: $m_{1}\times m_{2}\times\ldots \times m_{N}$: 
\[M = \Big(\langle a^{1}_{i_{1}i_{2}\ldots i_{N}}, a^{2}_{i_{1}i_{2}\ldots i_{N}}, \ldots a^{N}_{i_{1}i_{2}\ldots i_{N}} \rangle \Big)_{{i_{k} \leq m_{k}}\atop{k=1,2,\ldots N}}\] 
and 
\[\widehat{M} = \Big(\langle \hat{a}^{1}_{i_{1}i_{2}\ldots i_{N}}, \hat{a}^{2}_{i_{1}i_{2}\ldots i_{N}}, \ldots \hat{a}^{N}_{i_{1}i_{2}\ldots i_{N}} \rangle \Big)_{{i_{k} \leq m_{k}}\atop{k=1,2,\ldots N}}
\]

Let $c^{k}_{i_{1}i_{2}\ldots i_{N}} = \hat{a}^{k}_{i_{1}i_{2}\ldots i_{N}} - a^{k}_{i_{1}i_{2}\ldots i_{N}}$.

\begin{theorem} 
\label{GTN}
Let $M$, $\widehat{M}$ be $N$-person NF game matrices  of dimensions $m_{1}\times m_{2}\times\ldots \times m_{N}$. 
The matrix $\widehat{M}$ can be obtained from $M$ by an OI-transformation if and only if the following two conditions hold: 

\begin{enumerate}
\item[C$^{N}_{1}$:] $\sum_{j=1}^{N}  \hat{a}^{j}_{i_{1}i_{2}\ldots i_{N}} = \sum_{j=1}^{N} a^{j}_{i_{1}i_{2}\ldots i_{N}}$, for any ${i_{k} \leq m_{k}, \ k=1,2,\ldots N}$. 

(The sum of all payoffs in any given outcome must remain the same.)
 
\item[C$^{N}_{2}$:]  For any fixed $k,j=1,2,\ldots N$ the difference 
\[d^{j,k}_{i_{1}\ldots i_{N}}:= c^{j}_{i_{1}\ldots i_{k}+1\ldots i_{N}} - c^{j}_{i_{1}\ldots i_{k}\ldots i_{N}}\] 
is the same for every 
$i_{1},\ldots, i_{j}, \ldots, i_{N}$ 
such that $i_{p} \leq m_{p}, \ p=1,\ldots, N$ and 
$i_{j} < m_{j}$.
\end{enumerate}
\end{theorem} 

\begin{proof}

The proofs follows the same reasoning as in the 2-person case.

The necessity of condition C$^{N}_{1}$ is  obvious. 
To show the necessity of condition C$^{N}_{2}$, we do again simple calculations. Let 
$\widehat{M}$ be obtained from $M$ by an OI-transformation effected by side payment offers $\{ \delta^{i,k}_{j} \}_{k = 1,\ldots, N; 1 \leq j \leq m_{j}}$, where $\delta^{i,k}_{j}$ is the side payment offered by player $i$ to player 
$k\neq i$ contingent on $k$ playing action $j$. For technical convenience we also put $\delta^{i,i}_{j}=0$ for any $j$. Then we have: 
\[ c^{j}_{i_{1}\ldots i_{j}\ldots i_{N}} = \hat{a}^{j}_{i_{1}\ldots i_{j}\ldots i_{N}} - a^{j}_{i_{1}\ldots i_{j}\ldots i_{N}} =
\sum_{i=1}^{N} \delta^{i,j}_{i_{j}} - \sum_{k=1}^{N} \delta^{j,k}_{i_{k}}.
 \] 
Note that it suffices to show that 
\[d^{j,k}_{i_{1}\ldots i_{p}\ldots i_{N}} = d^{j,k}_{i_{1}\ldots i_{p}+1 \ldots i_{N}}\]
for any $j,k,p =1,\ldots, N$, i.e., 
\[c^{j}_{i_{1}\ldots i_{k}+1 \ldots i_{p} \ldots i_{N}} - c^{j}_{i_{1}\ldots i_{k} \ldots i_{p} \ldots i_{N}} =
c^{j}_{i_{1}\ldots i_{k}+1\ldots i_{p}+1\ldots i_{N}} - c^{j}_{i_{1}\ldots i_{k}\ldots i_{p}+1 \ldots i_{N}} 
\]
for $k < p \leq N$ and likewise
\[c^{j}_{i_{1}\ldots i_{p} \ldots i_{k}+1 \ldots i_{N}} - c^{j}_{i_{1}\ldots i_{p} \ldots i_{k} \ldots i_{N}} =
c^{j}_{i_{1}\ldots i_{p}+1 \ldots i_{k}+1\ldots i_{N}} - c^{j}_{i_{1}\ldots i_{p}+1 \ldots i_{k} \ldots i_{N}} 
\]
for $1 \leq p < k$. Both cases are completely analogous, so let us check the first equality. It is equivalent to 
\[c^{j}_{i_{1}\ldots i_{k}+1 \ldots i_{p} \ldots i_{N}} + 
c^{j}_{i_{1}\ldots i_{k}\ldots i_{p}+1 \ldots i_{N}} =
c^{j}_{i_{1}\ldots i_{k}+1\ldots i_{p}+1\ldots i_{N}} +
c^{j}_{i_{1}\ldots i_{k} \ldots i_{p} \ldots i_{N}} 
\]

First, suppose $j \neq p,k$. 
By definition, we have for the left hand side: 

$c^{j}_{i_{1}\ldots i_{k}+1 \ldots i_{p} \ldots i_{N}} + 
c^{j}_{i_{1}\ldots i_{k}\ldots i_{p}+1 \ldots i_{N}} =$

$\sum_{{q=1}}^{N} \delta^{q,j}_{i_{j}} - \sum_{{q=1}\atop{q \neq k}}^{N} \delta^{j,q}_{i_{q}} - \delta^{j,k}_{i_{k}+1}$ + 
$\sum_{{q=1}}^{N} \delta^{q,j}_{i_{j}} - \sum_{{q=1}\atop{q \neq p}}^{N} \delta^{j,q}_{i_{q}} - \delta^{j,p}_{i_{p}+1}$.

\medskip
Respectively, for the right hand side: 
\medskip

$c^{j}_{i_{1}\ldots i_{k}+1\ldots i_{p}+1\ldots i_{N}} +
c^{j}_{i_{1}\ldots i_{k} \ldots i_{p} \ldots i_{N}} =$
\medskip

$\sum_{{q=1}}^{N} \delta^{q,j}_{i_{j}} - \sum_{{q=1}\atop{q \neq k,p}}^{N} \delta^{j,q}_{i_{q}} - \delta^{j,p}_{i_{p}+1} - \delta^{j,k}_{i_{k}+1} + 
\sum_{{q=1}}^{N} \delta^{q,j}_{i_{j}} - \sum_{{q=1}}^{N} \delta^{j,q}_{i_{q}}$.

A direct inspection shows that these are equal. 

\medskip
Now, consider the case where $j=p$.
For the left hand side we get: 
\medskip

$c^{j}_{i_{1}\ldots i_{k}+1 \ldots i_{j} \ldots i_{N}} + 
c^{j}_{i_{1}\ldots i_{k}\ldots i_{j}+1 \ldots i_{N}} =$
\medskip

$\sum_{{q=1}}^{N} \delta^{q,j}_{i_{j}} - \sum_{{q=1}\atop{q \neq k}}^{N} \delta^{j,q}_{i_{q}} - \delta^{j,k}_{i_{k}+1}$ + 
$\sum_{{q=1}}^{N} \delta^{q,j}_{i_{j}+1} - \sum_{{q=1}\atop{q \neq j}}^{N} \delta^{j,q}_{i_{q}} - \delta^{j,j}_{i_{j}+1}$.

\medskip
Respectively, for the right hand side: 
\medskip

$c^{j}_{i_{1}\ldots i_{k}+1\ldots i_{j}+1\ldots i_{N}} +
c^{j}_{i_{1}\ldots i_{k} \ldots i_{j} \ldots i_{N}} =$
\medskip

$\sum_{{q=1}}^{N} \delta^{q,j}_{i_{j}+1} - \sum_{{q=1}\atop{q \neq k,j}}^{N} \delta^{j,q}_{i_{q}} - \delta^{j,j}_{i_{j}+1} - \delta^{j,k}_{i_{k}+1} + 
\sum_{{q=1}}^{N} \delta^{q,j}_{i_{j}} - \sum_{{q=1}}^{N} \delta^{j,q}_{i_{q}}$.

\medskip
Again, by direct inspection we see that these are equal. 

\medskip
Finally, the case where $j=k$ is completely analogous. 

\medskip
Now, suppose conditions C$^{N}_{1}$ and C$^{N}_{2}$ hold for the matrices $M$ and $\widehat{M}$, and all parameters $c^{j}_{i_{1}\ldots i_{j}\ldots i_{N}}$ are defined as before.
As in the 2-person case, we can show that the system of equations
\[ \Bigg\|
\sum_{i=1}^{N} \delta^{i,j}_{i_{j}} - \sum_{k=1}^{N} \delta^{j,k}_{i_{k}} = c^{j}_{i_{1}\ldots i_{j}\ldots i_{N}} 
\]
for all 
$j, i_{1},\ldots, i_{j}, \ldots, i_{N}$ 
such that  $i_{p} \leq m_{p}, \ p=1,\ldots, N$,  
 for the unknown real payments $\{ \delta^{i,k}_{j} \}_{k = 1,\ldots, N; 1 \leq j \leq m_{j}}$, 
is consistent. 
We can use again standard linear algebra and show that elementary matrix transformations would reduce the system to a consistent one in a canonical form. We omit the routine, but messy technicalities. 
\end{proof}

\medskip
As in the 2-person case, every OI-transformation can be determined by the payoffs for all players in the outcomes along the rows in all coordinate directions passing from any fixed outcome, that is, all outcomes resulting from all but one players following a fixed strategy profile. Here is the formal result. 

\begin{theorem} 
\label{GTaN}
Let 
\[M = \Big(\langle a^{1}_{i_{1}i_{2}\ldots i_{N}}, a^{2}_{i_{1}i_{2}\ldots i_{N}}, \ldots a^{N}_{i_{1}i_{2}\ldots i_{N}} \rangle \Big)_{{i_{k} \leq m_{k}}\atop{k=1,2,\ldots N}}
\]
and 
$\langle p_{1} \ldots p_{N}\rangle$ be a fixed tuple such that $p_{k}  \in 
\{1,\ldots, m_{k}\}$ for each $k=1,\ldots N$. Then every tuple\footnote{Note that each $\hat{a}^{k}_{p_{1}p_{2}\ldots p_{N}}$ occurs $N$ times in this list.}   \\ \\
$\mathcal{T} =  \bigg< \langle \hat{a}^{k}_{1p_{2}\ldots p_{N}}, \hat{a}^{k}_{2p_{2}\ldots p_{N}}, \ldots, \hat{a}^{k}_{m_{1}p_{2}\ldots p_{N}}$,  
$\hat{a}^{k}_{p_{1}1\ldots i_{N}}, \hat{a}^{k}_{p_{1}2\ldots i_{N}}, \ldots, \hat{a}^{k}_{p_{1}m_{2}\ldots p_{N}},  \ldots, $ \\  
$\hat{a}^{k}_{p_{1}p_{2}\ldots 1}, \hat{a}^{k}_{p_{1}p_{2}\ldots 2}, \ldots, \hat{a}^{k}_{p_{1}p_{2}\ldots m_{N}} \rangle \mid k = 1,\ldots, N \bigg>$ \\
\\ \\
satisfying condition C$^{N}_{1}$ of Theorem \ref{GTN} can be extended to a unique payoff matrix 
\[\widehat{M} = \Big(\langle \hat{a}^{1}_{i_{1}i_{2}\ldots i_{N}}, \hat{a}^{2}_{i_{1}i_{2}\ldots i_{N}}, \ldots \hat{a}^{N}_{i_{1}i_{2}\ldots i_{N}} \rangle \Big)_{{i_{k} \leq m_{k}} \atop {k=1,2,\ldots N}}\] 
that can be obtained from $M$ by an OI-transformation. 
\end{theorem} 

\begin{proof} Given any tuple of values  $\mathcal{T}$  satisfying the conditions of the theorem, the extension to a matrix $\widehat{M}$ can be done  as follows. 
First, we use $\mathcal{T}$ and the entries of $M$ to compute the values of all $c^{j}_{i_{1}\ldots  i_{N}}$ whose index vectors  corresponding to the entries in $\mathcal{T}$.

We take the identities 
\[d^{j,k}_{i_{1}\ldots i_{p}\ldots i_{N}} = d^{j,k}_{i_{1}\ldots i_{p}+1 \ldots i_{N}}
\mbox{ for all } j,k,p =1,\ldots, N\] 
and expand them (assuming e.g., that $k < p \leq N$):  
\[c^{j}_{i_{1}\ldots i_{k}+1 \ldots i_{p} \ldots i_{N}} - c^{j}_{i_{1}\ldots i_{k} \ldots i_{p} \ldots i_{N}} =
c^{j}_{i_{1}\ldots i_{k}+1\ldots i_{p}+1\ldots i_{N}} - c^{j}_{i_{1}\ldots i_{k}\ldots i_{p}+1 \ldots i_{N}} 
\]
Then we rewrite them as: 
\[
c^{j}_{i_{1}\ldots i_{k}+1\ldots i_{p}+1\ldots i_{N}} 
=  c^{j}_{i_{1}\ldots i_{k}\ldots i_{p}+1 \ldots i_{N}} 
+ c^{j}_{i_{1}\ldots i_{k}+1 \ldots i_{p} \ldots i_{N}} - 
c^{j}_{i_{1}\ldots i_{k} \ldots i_{p} \ldots i_{N}}
\]
These are recurrent formulae computing the values of all $c^{j}_{i_{1}\ldots i_{N}}$ from those computed initially by propagating from $c^{j}_{p_{1}\ldots p_{N}}$ in all  diagonal directions. We leave the tedious details out.  
Once all $c^{j}_{i_{1}\ldots i_{N}}$ are computed, the matrix $\widehat{M}$ is determined. By construction it satisfies  conditions C$^{N}_{1}$ and C$^{N}_{2}$  of  Theorem \ref{GTN}, hence it can be obtained from $M$ by an OI-transformation. 
 The uniqueness of $\widehat{M}$ follows immediately. 
\end{proof}

\begin{example}
Consider a 3-person game with players $A_{1},A_{2},A_{3}$ of dimensions $3 \times 3 \times 2$, with 2-dimensional matrix-slices for the 2 actions of player $A_{3}$  as follows: 
\[ M_{1} = 
 \left[\begin{array}{c|ccc}
A_{1}\backslash  A_{2} & A_{211} & A_{221} & A_{231}  \\
\hline 
A_{111} &  1,2,0 &  2,3,1 &  3,1,2  \\
A_{121} &  2,3,3 & 3,4,4 &  4,2,5  \\ 
A_{131} &  6,5,6 & 7,6,7 &  5,7,8  \\ 
\end{array}\right], 
\ \ \ 
 M_{2} = 
 \left[\begin{array}{c|ccc}
A_{1}\backslash  A_{2} & A_{212} & A_{222} & A_{232}  \\
\hline 
A_{112} &  1,1,8 &  2,2,7 &  3,3,6  \\
A_{122} &  1,2,5 & 2,3,4 &  3,4,3  \\ 
A_{132} &  2,1,2 & 3,2,1 &  1,3,0  \\ 
\end{array}\right]
\]
Suppose the tuple of entries for the transformed matrix is given in terms of the outcome $(2,2,1)$ as follows: 
\[ \widehat{M}^{0}_{1} = 
 \left[\begin{array}{c|ccc}
A_{1}\backslash  A_{2} & A_{211} & A_{221} & A_{231}  \\
\hline 
A_{111} &               &  1,2,3 &            \\
A_{121} &   4,4,0   &  5,1,5 &  3,4,4  \\ 
A_{131} &               &  8,4,8 &             \\ 
\end{array}\right], 
\ \ \ 
 \widehat{M}^{0}_{2} = 
 \left[\begin{array}{c|ccc}
A_{1}\backslash  A_{2} & A_{212} & A_{222} & A_{232}  \\
\hline 
A_{112} &   &   &    \\
A_{122} &  & 3,3,3 &    \\ 
A_{132} &   &  &    \\ 
\end{array}\right]
\]
Note that condition C$^{N}_{1}$ of Theorem \ref{GTN} is satisfied. 

The corresponding partial slices of the matrix of differences 
$ c^{j}_{i_{1} i_{2} i_{3}} = \hat{a}^{j}_{i_{1} i_{2} i_{3}} - a^{j}_{i_{1} i_{2} i_{3}}$ 
are: 
\[ C^{0}_{1} = 
 \left[\begin{array}{c|ccc}
A_{1}\backslash  A_{2} & A_{211} & A_{221} & A_{231}  \\
\hline 
A_{111} &              &  -1,-1,2   &    \\
A_{121} &  2,1,-3  &    2,-3,1   &  -1,2,-1  \\ 
A_{131} &              &    1,-2,1   &   \\ 
\end{array}\right], 
\ \ \ 
 C^{0}_{2} = 
 \left[\begin{array}{c|ccc}
A_{1}\backslash  A_{2} & A_{212} & A_{222} & A_{232}  \\
\hline 
A_{112} &   &   &    \\
A_{122} &  & 1,0,-1 &    \\ 
A_{132} &   &  &    \\ 
\end{array}\right]
\]

The remaining entries of that matrix are then computed 
consecutively by using the identities 
\[
c^{j}_{i_{1}+1 i_{2}+1 i_{3}} 
=  c^{j}_{i_{1}+1 i_{2} i_{3}} + c^{j}_{i_{1} i_{2}+1 i_{3}}  - 
c^{j}_{i_{1} i_{2} i_{3}} 
\]
\[
c^{j}_{i_{1} i_{2}+1 i_{3}+1} 
=  c^{j}_{i_{1} i_{2} i_{3}+1} + c^{j}_{i_{1} i_{2}+1 i_{3}}  - 
c^{j}_{i_{1} i_{2} i_{3}},  
\]
\[
c^{j}_{i_{1}+1 i_{2} i_{3}+1} 
=  c^{j}_{i_{1} i_{2} i_{3}+1} + c^{j}_{i_{1}+1 i_{2} i_{3}}  - 
c^{j}_{i_{1} i_{2} i_{3}},
\]
for each $j = 1,2,3$.
Thus, we first obtain 
\[ C_{1} = 
 \left[\begin{array}{c|ccc}
A_{1}\backslash  A_{2} & A_{211} & A_{221} & A_{231}  \\
\hline 
A_{111} &    -1,3,-2  &  -1,-1,2   &  -4,4,0  \\
A_{121} &  2,1,-3    &    2,-3,1   &  -1,2,-1  \\ 
A_{131} &  1,2,-3    &    1,-2,1   &   -2,3,-1 \\ 
\end{array}\right], 
\]
then: 
\[
 C^{1}_{2} = 
 \left[\begin{array}{c|ccc}
A_{1}\backslash  A_{2} & A_{212} & A_{222} & A_{232}  \\
\hline 
A_{112} &             &  -2,2,0      &    \\
A_{122} &  1,4,-5 & 1,0,-1       & -2,5,-3   \\ 
A_{132} &             &  0,1,-1      &    \\ 
\end{array}\right]
\]
and finally: 
\[
 C_{2} = 
 \left[\begin{array}{c|ccc}
A_{1}\backslash  A_{2} & A_{212} & A_{222} & A_{232}  \\
\hline 
A_{112} & -2,6,-4      &  -2,2,0      &  -5,7,-2   \\
A_{122} &  1,4,-5      & 1,0,-1       & -2,5,-3   \\ 
A_{132} &   0,5,-5     &  0,1,-1      &  -3,6,-3 \\ 
\end{array}\right]
\]
Eventually, we obtain the 2-dimensional matrix-slices of $\widehat{M}$ for the 2 actions of player $A_{3}$:  
\[ \widehat{M}_{1}  = 
 \left[\begin{array}{c|ccc}
A_{1}\backslash  A_{2} & A_{211} & A_{221} & A_{231}  \\
\hline 
A_{111} & 0,5,-2      &  1,2,3      &  -1,5,2   \\
A_{121} &  4,4,0      & 5,1,5       &    3,4,4   \\ 
A_{131} &   7,7,3     &  8,4,8      &  3,10,7 \\ 
\end{array}\right]
\]
\[ \widehat{M}_{2}  = 
 \left[\begin{array}{c|ccc}
A_{1}\backslash  A_{2} & A_{212} & A_{222} & A_{232}  \\
\hline 
A_{112} & -1,7,4      &  0,4,7      &  -2,10,4   \\
A_{122} &  2,6,0      & 3,3,3       &   1,9,0   \\ 
A_{132} &   2,6,-3    &  3,3,0      &  -2,9,-3 \\ 
\end{array}\right]
\]
\end{example}

\medskip
In summary, as in the case of 2-person games, theorems \ref{GTN} and \ref{GTaN} together characterize precisely the payoff matrices $\widehat{M}$ that can be obtained  from matrix $M$ by an OI-transformation: any of them can be obtained by choosing any strategy profile in $M$ and setting suitable payoffs satisfying condition C$^{N}_{1}$ for all outcomes obtained by allowing any one, and only one, player to deviate from that strategy profile. Then all payoffs in $\widehat{M}$ are computed by using the recurrent formulae derived from condition C$^{N}_{2}$.

\section{Concluding remarks}
As already stated, the contributions of the present paper are purely technical: explicit and easy to apply and use characterizations of the game matrix transformations that can be induced by preplay offers for payments or threats for punishments in normal form games. Even though we have not considered rationality issues that would determine which of these transformations can be effected by offers made by \emph{rational players}, we believe that our results are of direct game-theoretic relevance, because they can be used by the players to determine what mutually desirable transformed games (e.g., having dominant stategy equilibria with Pareto optimal outcomes) they can achieve by exchange of preplay offers, and then to search -- by using the computational procedures that can be extracted from our proofs -- for suitable offers that would induce the necessary game matrix transformations leading to the desired outcomes. 

\bibliographystyle{alpha}
\bibliography{Preplay-Bibliography}

\end{document}